\documentclass[11pt]{article}       
\usepackage{geometry}               
\usepackage{graphicx}
\usepackage{caption}
\usepackage{subcaption}
\usepackage{comment}
\usepackage{color}
\usepackage{amsmath} 
\usepackage{amssymb}  
\usepackage{amsthm}  
\usepackage{bm}
\usepackage{lipsum}
\usepackage[linesnumbered, ruled]{algorithm2e}
\usepackage{color}

\newtheorem{proposition}{Proposition}

\newtheorem{theorem}{Theorem}
\newtheorem{remark}{Remark}

\newcommand{\food}{\mbox{\tiny food}}
\newcommand{\air}{\mbox{\tiny air}}
\newcommand{\wall}{\mbox{\tiny wall}}
\newcommand{\evap}{\mbox{\tiny evap}}
\newcommand{\amb}{\mbox{\tiny amb}}
\newcommand{\suc}{\mbox{\tiny suc}}
\newcommand{\meas}{\mbox{\tiny meas}}

\numberwithin{equation}{section}

\usepackage{lipsum}
\usepackage{marginnote,zref-savepos}
\newcounter{labelnote}
\makeatletter
\let\oldmarginnote\marginnote
\renewcommand*{\marginnote}[1]{%
 \begingroup\strut
  \stepcounter{labelnote}\zsaveposx {marginnote-\thelabelnote}
     \ifnum 0\zposx{marginnote-\thelabelnote}<19000000
      \reversemarginpar
      \oldmarginnote{\color{blue}#1}%
     \else
      \normalmarginpar
      \oldmarginnote{\color{blue}#1}%
     \fi
 \endgroup%
}
\makeatother

\title{Online Combinatorial Optimization for Interconnected Refrigeration Systems: Linear Approximation and Submodularity
 } 

\author{
 Insoon Yang\thanks{Ming Hsieh Department of Electrical Engineering, University of Southern California ({insoonya@usc.edu}). Supported in part by NSF under CRII:CPS (CNS1657100).}
}

\date{}

\begin{document}
\maketitle


\pagestyle{myheadings}
\thispagestyle{plain}

\begin{abstract}
Commercial refrigeration systems consume 7\% of the total commercial energy consumption in the United States. 
Improving their energy efficiency contributes to
the sustainability of global energy systems and the supermarket business sector.
This paper proposes a new control method that can save the energy consumption of multi-case supermarket refrigerators by explicitly taking into account their interconnected and switched system dynamics.
Its novelty is a bilevel combinatorial optimization formulation to generate ON/OFF control actions for expansion valves and compressors.
The inner optimization module keeps display case temperatures in a desirable range and the outer optimization module minimizes energy consumption.
In addition to its energy-saving capability, the proposed controller significantly reduces the frequency of compressor switchings by employing a conservative compressor control strategy.
However, solving this bilevel optimization problem associated with interconnected and switched systems is a computationally challenging task.
To solve the problem in near real time, we propose two approximation algorithms that can solve both the inner and outer optimization problems at once.
The first algorithm uses a linear approximation, and the second is based on the submodular structure of the optimization problem.
Both are (polynomial-time) scalable algorithms and generate near-optimal solutions with performance guarantees. 
Our work complements  existing optimization-based control methods (e.g., MPC) for supermarket refrigerators, as our algorithms can be adopted as a tool for solving combinatorial optimization problems arising in these methods.
\end{abstract}


\section{Introduction}

Commercial refrigeration systems account for 7\% of the total commercial energy consumption in the United States \cite{DOE2012}.
Therefore, there is a strong need for energy-efficient refrigeration systems, but research and development have focused on improving hardware rather than software, including control systems.
Traditionally, hysteresis and set point-based controllers have been used to maintain the display case temperature in a desirable range without considering system dynamics and energy consumption.
Over the past decade, however, more advanced control systems have been developed to save energy consumption
using real-time sensor measurements and optimization algorithms (see Section \ref{lit}).
Advances in new technologies, such as the Internet of Things and cyber-physical systems, enhance the practicality of 
such an advanced control system with their sensing, communication, and computing capabilities \cite{Graziano2014}.

Supermarkets are one of the most important commercial sectors in which energy-efficient refrigeration systems are needed.
The primary reasons are twofold.  First, supermarket refrigerators consume 56\% of energy consumed by commercial refrigeration systems \cite{Navigant2009}.
Second, supermarkets  operate with very thin profit margins (on the order of 1\%), and energy savings thus significantly help their business: the U.S. Environmental Protection Agency estimates that reducing energy costs by \$1 is equivalent to increasing sales by \$59 \cite{ES2008}.
However, improving the energy efficiency of supermarket refrigerators is a challenging task because food products must be stored at proper temperatures.
Failure to do so will increase food safety risks.
The most popular refrigerators in supermarkets are multi-display case units. 
An example is illustrated in 
Fig. \ref{fig:diagram_int}.
Each display case has an evaporator controlled by an expansion valve, and a unit's suction pressure is controlled by a compressor rack, as shown in Fig. \ref{fig:diagram}.
In typical supermarket refrigerators,  controllers turn ON and OFF expansion valves and compressors to keep display case temperatures in a specific range.
Importantly, there are heat transfers between display cases due to the interconnection among them.
Note that traditional hysteresis or set point-based controllers do not take into account such heat transfers and therefore perform in a suboptimal way.

\begin{figure}[tb] 
\begin{center}
\includegraphics[width =3.3in]{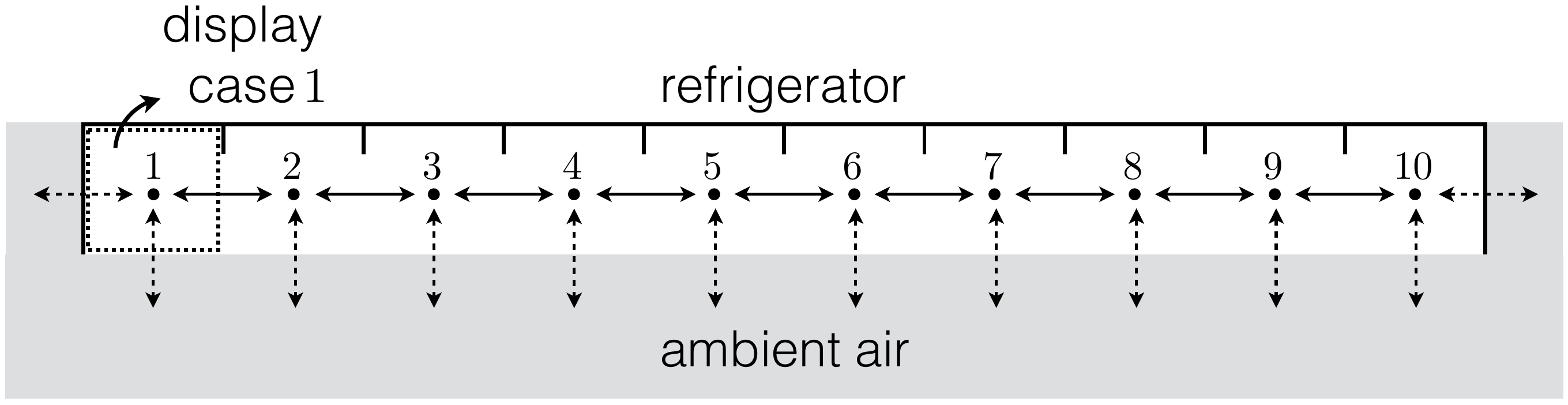}
\caption{A supermarket refrigerator, which has 10 display cases.
Evaporator $i$ controls the temperature of display case $i$.
Lines with arrows represent heat transfers between neighboring display cases or between a display case and ambient air.}
 \label{fig:diagram_int}
 \end{center}
\end{figure}

This paper proposes a new control method that can improve the energy efficiency of multi-case supermarket refrigerators by explicitly taking into account the interconnected and switched dynamics of display case temperatures. 
The proposed controller receives sensor measurements and optimizes ON/OFF control actions for expansion valves and compressors in near real time. 
The novelty of this work is a bilevel combinatorial optimization formulation to generate such ON/OFF control signals in which $(i)$ the inner combinatorial optimization module 
is responsible for maintaining display case temperatures in a desirable range, and $(ii)$ the outer combinatorial optimization module minimizes energy consumption.
The primary advantage of the proposed approach is its energy savings. 
Because the controller explicitly takes into account the system dynamics and heat transfers,
it effectively uses state measurements and optimizes control actions to save energy while guaranteeing desired temperature profiles.
In our case studies, the proposed control method saves 7.5--8$\%$ of energy compared to a traditional approach.
The secondary benefit of the proposed method is to reduce the frequency of compressor switchings.
It is known that frequent switchings of compressors accelerate their mechanical wear. 
We propose a conservative compressor control approach that reduces fluctuations in suction pressure and thus decreases the compressor switching frequency.  In our case studies using a benchmark refrigeration system model, the proposed method reduces the switching frequency by 54--71.6$\%$.

\begin{figure}[tb] 
\begin{center}
\includegraphics[width =2.7in]{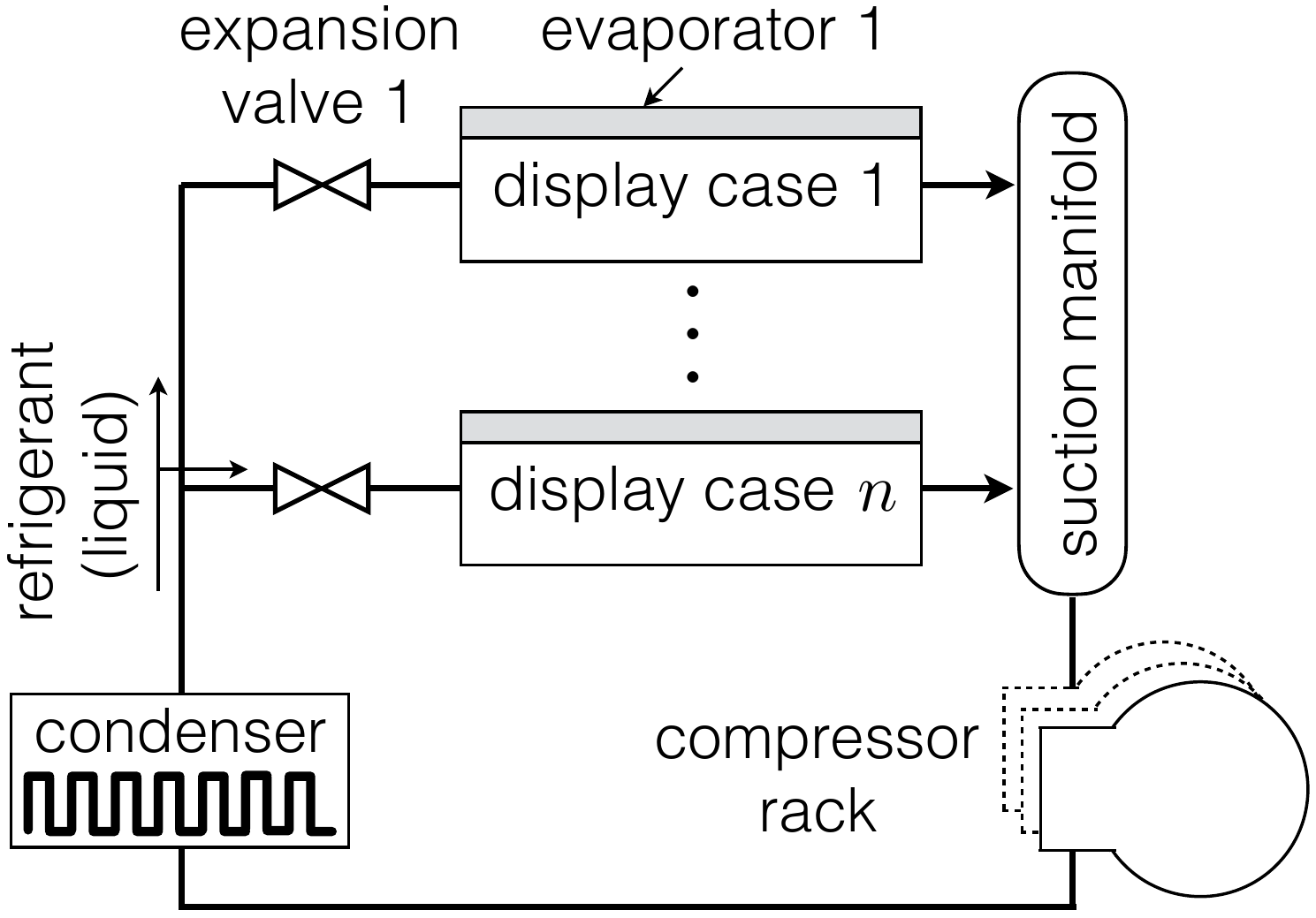}
\caption{Schematic diagram of a supermarket refrigerator.}
 \label{fig:diagram}
 \end{center}
\end{figure}

The proposed control method, however, presents a theoretical and algorithmic challenge 
because a bilevel combinatorial optimization associated with a dynamical system must be solved in near real time. 
To overcome this challenge, we suggest  two approximation algorithms that can solve both of the inner and outer optimization problems at once. 
The first algorithm uses the linear approximation method developed in our previous work \cite{Yang2016}.
The approximate problem is a linear binary program, which can be solved by an efficient and scalable single-pass algorithm.
In addition, it simulates the dynamical system model only once to generate control actions at each time point. 
We also show that the approximate solution obtained by this method has a provable suboptimality bound.
The second algorithm is based on the \emph{submodular} structure in the optimization problem.
The inner optimization's objective function is submodular because opening an expansion valve when a smaller set of valves are opened gives a greater marginal benefit than opening it when a larger set of valves are already opened.
We prove this intuitive submodularity property.  Therefore, a greedy algorithm can be adopted to obtain a $(1-\frac{1}{e})$-optimal solution \cite{Nemhauser1978}.
In our case studies, the actual performance of the proposed controller using these two algorithms is $98.9$--$99.5$\% of the optimal controller.

\subsection{Related Work}\label{lit}

Several optimization-based control methods for commercial refrigerators have been developed over the past decade. 
One of the most popular methods is model predictive control (MPC) although it is computationally challenging to apply standard MPC due to the  switched  dynamics of refrigeration systems.
It is shown that the mixed logical dynamical framework is useful to solve 
small-size problems with a piecewise affine approximation of a system model \cite{Bemporad1999, Larsen2005}.
However, the practicality of this method is questionable due to the high dimensionality of practical problems for supermarket refrigerators, except for limited cases.
To overcome this limitation, \cite{Sarabia2009} carefully selects and parametrizes optimization variables to formulate the problem as nonlinear MPC instead of hybrid MPC.
Nonetheless this approach is computationally expensive because a nonlinear program with many variables must be solved in each MPC iteration.
An alternative approach using hierarchical MPC is proposed in \cite{Sonntag2008}.
This method separates time scales into two: in every  nonlinear MPC iteration, low-level temperature controllers were employed, and the high-level optimization task is to determine optimal parameters for these controllers.
However, this approach still presents the combinatorial growth of the search space.
More recently, a sequential convex programming-based method is shown to be computationally efficient in several case studies \cite{Hovgaard2013}.
It iteratively solves an optimization problem using convex programming, replacing the nonconvex cost function with a convex approximation.
In several numerical experiments, this heuristic method generates high-quality control signals although it gives no theoretical performance guarantee.
We believe that our work is complementary to the aforementioned methods. 
One of our main contributions is to develop two efficient and scalable algorithms for resolving the computational challenge in discrete optimization problems associated with supermarket refrigeration systems.
These algorithms can be adopted as a tool for solving combinatorial optimization problems in the aforementioned methods.
We propose one of the most efficient control architectures that use the algorithms. 

With  advances in modern power systems, 
an important emerging application is using
 supermarket refrigeration systems for thermal storage \cite{Pedersen2013, Shafiei2013, Vinther2015, Minko2015}.
In this application, it is often necessary to regulate the total power consumption of hundreds of refrigerators to provide reliable demand response services to the power grid. 
Our work contributes to such demand response applications by providing scalable online optimization algorithms with performance guarantees that are useful for solving large-scale problems.
The utility of the proposed algorithms in demand response is demonstrated in \cite{Yang2015} and a case study is presented in Section \ref{case}.

\subsection{Outline}

The remainder of this paper is organized as follows.
In Section \ref{ref_model}, we describe an interconnected hybrid system model of supermarket refrigerators and provide a simulation result when a traditional set point- and PI-based controller is employed.
We then propose the proposed control method based on bilevel online combinatorial optimization in Section \ref{control}.
In Section \ref{app_alg}, we provide two efficient algorithms to solve the combinatorial optimization problem near real time and examine their scalability.
In Section \ref{case}, we compare the performance of the proposed controllers with that of the traditional controller and demonstrate its utility in automated demand response.

\section{Switched and Interconnected Dynamics of Supermarket Refrigeration Systems}\label{ref_model}

We consider a supermarket refrigerator in which multiple display cases are interconnected with one another.
For example, Fig. \ref{fig:diagram_int} shows a refrigerator that has 10 display cases. 
The temperature of each display case is controlled by an evaporator unit, where the refrigerant evaporates absorbing heat from the display case.
Let evaporator $i$ be in charge of display case $i$ for $i = 1, \cdots, n$, where $n$ is the number of display cases in all the refrigerators.
Several dynamic models of supermarket refrigeration systems 
have been proposed 
\cite{Rasmussen2006, Larsen2007, Li2010, Rasmussen2012, Rasmussen2012b, Shafiei2013b} (see also the references therein). 
Among those, we use the benchmark model of a typical supermarket refrigeration system proposed in \cite{Larsen2007} and 
 widely used in \cite{Sarabia2009, Sonntag2008, Yang2011, Vinther2015, Minko2015}.
This model is useful for simulating display case temperatures and evaluating the performances of several controllers.

\subsection{Display Cases and Evaporators}

Display cases store food products and keep them refrigerated.
This refrigeration is due to the heat transfer between the food product and the cold air in the display cases.
Let $T_{\food, i}$ and $T_{\air,i}$ denote the temperatures of the food product and the air in display case $i$.
The heat transfer $Q_{\food \to \air, i}$ between the food product and the air in display case $i$ can then be modeled as
\begin{equation}\label{food_dynamics}
\begin{split}
m_{\food, i} c_{\food, i} \dot{T}_{\food, i} &= - Q_{\food \to \air, i}\\
&=-k_{\food-\air} (T_{\food,i} - T_{\air,i}),
\end{split}
\end{equation}
where $m_{\food, i}$ is the mass of the food product, $c_{\food,i}$ is the heat capacity of the food product and $k_{\food-\air}$ is the heat transfer coefficient between the food product and the air.

The display case air temperature is affected by the heat transfers from the food product $(Q_{\food \to \air, i})$, the ambient air $(Q_{\amb \to \air, i})$, the evaporator  $(-Q_{\air \to \evap, i})$ and the neighboring display case air $(\sum_{j = 1}^n Q_{j \to i})$. 
The refrigerant flow into an evaporator is controlled by its expansion valve.  
Let $u_i$ be the valve control variable for evaporator $i$ such that
\begin{equation} \nonumber
u_i(t) := \left \{
\begin{array}{ll}
0 &\mbox{if expansion valve $i$ is closed at $t$}\\
1 &\mbox{otherwise}.
\end{array}
\right.
\end{equation}
Expansion valve $i$ controls the refrigerant injection into evaporator $i$
and decreases the pressure of the refrigerant if it is open, as shown in Fig. \ref{fig:diagram}.
Then, the dynamics of the display case air temperature can be modeled as
the following switched interconnected system:
\begin{equation}\label{air_dynamics}
\begin{split}
&m_{\air,i}c_{\air, i} \dot{T}_{\air, i} \\
&= Q_{\food \to \air, i} + Q_{\amb \to \air, i}
- Q_{\air \to \wall, i} + \sum_{j = 1}^n Q_{j \to i}\\
&=k_{\food-\air}(T_{\food, i} - T_{\air, i}) + k_{\amb - \air} ( T_{\amb} - T_{\air, i})\\
&- k_{\air-\evap} (T_{\air,i} - T_{\evap} u_i) + \sum_{j=1}^n k_{i,j} (T_{\air, j} - T_{\air,i}),
\end{split}
\end{equation}
where $T_{\amb}$ is the ambient air temperature, $T_{\evap}$ is the refrigerant's evaporation temperature,
$k_{\amb-\air}$ is the heat transfer coefficient between the ambient air and the display case air and $k_{i,j}$ is the heat transfer coefficient between display case~$i$'s air and display case $j$'s air.
Note that $k_{i,j} = 0$ if display cases~$i$ and $j$ are not neighbors.
For a more detailed model, one can separately consider the dynamics of the evaporator wall temperature \cite{Sarabia2009}.\footnote{Alternatively, one can introduce a delay parameter, $\tau$, and replace $T_{\evap}u_i(t)$ with $T_{\evap} u_i(t-\tau)$ to explicitly take into account the effect of the evaporator wall temperature.} 
However, the proposed model is a good approximation because the heat transfer coefficient between the evaporator wall and the refrigerant is five to ten times higher than other heat transfer coefficients \cite{Larsen2007}.

The mass flow out of the evaporator can be computed as
\begin{equation}\nonumber
\begin{split}
f_i &:= \frac{1}{\Delta t} m_{\mbox{\tiny ref}, i},
\end{split}
\end{equation}
where the refrigerant mass in the evaporator is controlled by the valve switching
\begin{equation}\nonumber
{m}_{\mbox{\tiny ref}, i} = \left \{
\begin{array}{ll}
m_{\mbox{\tiny ref}}^{\max} & \mbox{if $u_{i} = 1$}\\
0 & \mbox{if $u_{i} = 0$}.
\end{array}
\right.
\end{equation}
Depending on the specification of refrigerators, 
it takes a nontrivial amount of time to fill up the evaporator by refrigerant.
In this case, the dynamics of the refrigerant mass in the evaporator can  be explicitly taken into account \cite{Sarabia2009}. Alternatively, one can introduce a delay-time constant, $\tau$, and let $m_{\mbox{\tiny ref}, i}(t) = m_{\mbox{\tiny ref}}^{\max} u_i(t-\tau)$ to model the effect of the time to fill up the evaporator.

\subsection{Suction Manifold and Compressor Rack}

As shown in Fig. \ref{fig:diagram}, 
the evaporated refrigerant with low pressure from the outlet of the evaporator is compressed by the electric motors in the compressor bank.
Each refrigerator could have multiple compressors and each compressor is switched ON or OFF. 
For example, all the compressors are turned ON when maximal compression is needed.
The compressor bank is conventionally controlled by a PI controller to maintain the suction pressure within a bandwidth.

The suction manifold pressure $P_{\suc}$ evolves with the following dynamics:
\begin{equation}\label{p_dy}
\dot{P}_{\suc} = \frac{1}{V_{\suc} r_{\suc}} \left (
\sum_{i=1}^n f_i - \rho_{\suc} \sum_{i=1}^{n_c} F_{c,i}
\right ),
\end{equation}
where $V_{\suc}$ is the volume of the suction manifold,
 $\rho_{\suc}$ is the density of the refrigerant in the suction manifold, and 
 $r_{\suc} := d\rho_{\suc}/dP_{\suc}$.
The variable $F_{c,i}$ denotes
the volume flow out of the suction manifold controlled by compressor $i$.
Let $u_{c,i}$ be the control variable for compressor $i$, where $u_{c,i} = 0$ represents that compressor $i$ is OFF and $u_{c,i} = 1$ represents that compressor $i$ is ON.  The volume flow $F_{c,i}$ is then given by
\begin{equation}\nonumber
F_{c,i} = k_{c} u_{c,i} := \frac{\eta V_{\mbox{\tiny comp}}}{n} u_{c,i},
\end{equation}
where $\eta$ is the volumetric efficiency of each compressor, and $V_{\mbox{\tiny comp}}$  denotes the compressor volume. 

The total power consumption by the compressor rack is given by
\begin{equation}\nonumber
p = \rho_{\suc} (h_{oc} - h_{ic}) \sum_{i=1}^{n_c} F_{c,i},
\end{equation}
where $h_{ic}$ and $h_{oc}$ are the enthalpies of the refrigerant flowing 
 into and out of
 the compressor, respectively.
 The compressed refrigerant  flows to the condenser and is liquefied by generating heat,  as shown in Fig. \ref{fig:diagram}.
The liquefied refrigerant flows to the expansion valve, and as a result, the refrigeration circuit is closed.

\subsection{Traditional Set-Point/PI-Based Control}\label{convention}

A widely used control method consists of $(i)$ a set-point based control of expansion valves, and $(ii)$ a PI control of compressors \cite{Larsen2007, Sarabia2009}.
Specifically, the following ON/OFF control law is traditionally used for expansion valve $i$:
\begin{equation} \nonumber
u_i (t) := \left \{
\begin{array}{ll}
1 & \mbox{if } T_{\air, i} (t) >T_i^{\max} \\
0 & \mbox{if } T_{\air, i} (t) <T_i^{\min} \\
u_i(t^-) & \mbox{otherwise},
\end{array}
\right.
\end{equation}
where $[T_i^{\min}, T_i^{\max}]$ is the desirable temperature range for display case $i$. 
To control the suction pressure, compressors are traditionally operated by a PI controller. 
This controller tracks the error $e(t)$ that measures the deviation of the suction pressure from the reference $\bar{P}_{\suc}$ over the dead band $DB$, i.e.,
\begin{equation}\nonumber
e(t) := \left \{
\begin{array} {ll}
P_{\suc} (t) - \bar{P}_{\suc} &\mbox{if } |e(t)| > DB\\
0 & \mbox{otherwise}.
\end{array}
\right.
\end{equation}
Then, the number of ON compressors is determined by a thresholding rule
depending on the following output of the PI controller
 \begin{equation}\nonumber
u_{PI} (t) = K_P e(t) +\frac{1}{K_I} \int e(t) dt;
\end{equation}
the greater value the output generates, the more compressors the controller turns on. More details about the thresholding rule can be found in \cite{Larsen2007, Sarabia2009}.

In our case studies,  R134a is chosen as the refrigerant.  Its relevant thermodynamic properties are contained in \cite{Larsen2007}.  For convenience, we summarize the properties as follows:  
$T_{\evap} = -4.3544P_{\suc}^2 + 29.2240 P_{\suc} - 51.2005$,
$\Delta h = (0.0217 P_{\suc}^2 - 0.1704 P_{\suc} + 2.2988) \times 10^5$, 
$\rho_{\suc} = 4.6073 P_{\suc} + 0.3798$,
$r_{\suc} = -0.0329 P_{\suc}^3 + 0.2161P_{\suc}^2 - 0.4742 P_{\suc} + 5.4817$,
$\rho_{\suc} (h_{oc} - h_{ic}) = (0.0265 P_{\suc}^3 - 0.4346 P_{\suc}^2 + 2.4923 P_{\suc} + 1.2189 ) \times 10^5$.
These formulas were obtained by fitting polynomials with experimental data.
The additional parameters used in the simulations are summarized in Table \ref{table:parameters}.

\begin{center} 
    \captionof{table}{Parameters used in simulations} \label{table:parameters}
    \small
  \begin{tabular}{   c c  || c c} 
    \hline
    Parameter & Value & Parameter & Value\\ \hline\hline
    $\bar{m}_{\food, i}$ & 200 kg & $c_{\food, i}$ & 1000 J/(kg$\cdot$K)  \\ \hline
    $m_{\wall, i}$ &  260 kg  & $c_{\wall, i}$ &  385 J/(kg$\cdot$K) \\ \hline
$m_{\air, i}$ &  50 kg  & $c_{\air, i}$ & 1000 J/(kg$\cdot$K) \\ \hline
$k_{\food-\air}$ &  300 J/(s$\cdot$K)  & $k_{\air-\evap}$ &  225 J/(s$\cdot$K)  \\ \hline
$k_{i,j}$  & 500 J/(s$\cdot$K)  & ${k}_{\amb-\evap}$ &  275 J/(s$\cdot$K)  \\ 
\hline
$m_{\mbox{\tiny ref}}^{\max}$ & 1 kg & $V_{\suc}$ & 10 m$^3$\\ \hline
 $\eta $ & 0.81 & $V_{\mbox{\tiny comp}}$ & 0.2 m$^3$/s\\ \hline
 $n$ &10  & $T_{\amb}$ & 20$^\circ$C\\ \hline
 ${T}_{i}^{\min} $ & 0$^\circ$C & ${T}_{i}^{\max}$ & 5$^\circ$C \\ \hline
  $K_P$ & $0.1$ & $K_I$ & $-0.8$\\ \hline
   $\bar{P}_{\suc}$ & 1.4 bar & $DB$ & 0.3 bar \\ \hline
  \end{tabular}
    \end{center}

We perturbed the mass of food products in each display case by $\pm 20\%$ from the nominal value $\bar{m}_{\food, i}$.
Despite this heterogeneity, the set point-based controller almost identically turns ON and OFF
all the expansion valves and therefore all the display case temperatures have almost the same trajectory as shown in Fig. \ref{fig:PI} (a).
This synchronization is due to the decentralized nature of the set point-based controller: the control decision for expansion valve $i$ depends only on its local temperature. 
Intuitively, this decentralized controller is suboptimal because it does not actively take into account the heat transfer between neighboring display cases.  This inefficiency of the traditional control approach motivates us to develop a new optimization-based control method that explicitly considers the interdependency of display case temperature dynamics.

Another disadvantage resulting from the synchronization of expansion valves is the significant fluctuation of suction pressure. Since the PI controller integrates the deviation of suction pressure from its reference, the output $u_{PI} (t)$ presents large and frequent variations.  
As a result, the number of ON compressors frequently varies as shown in 
Fig. \ref{fig:PI} (b). 
A frequent switching of compressors is a serious problem because it accelerates the mechanical degradation of the compressors.
Our strategy to discourage frequent compressor switchings is twofold:
$(i)$ 
our conservative compressor control method tries to maintain $P_{\suc} (t) = \bar{P}_{\suc}$, 
not fully utilizing the pressure bandwidth $\pm DB$, and 
$(ii)$ our online optimization-based controller 
indirectly desynchronizes the ON/OFF operation of expansion valves.
The details about the two control methods are presented in the following sections.
Unlike the traditional control approach, our proposed method 
is suitable for regulating the total power consumption in real time.
This feature is
 ideal for modern power system (so-called `smart grid') applications,
 allowing supermarket units to follow a real-time regulation signal 
 for reducing peak demand or supporting a spinning reserve.
Such applications of our control method to power systems are studied in \cite{Yang2015} and one of which is studied in Section \ref{dr}.

\begin{figure}
\centering
  \includegraphics[width=3.5in]{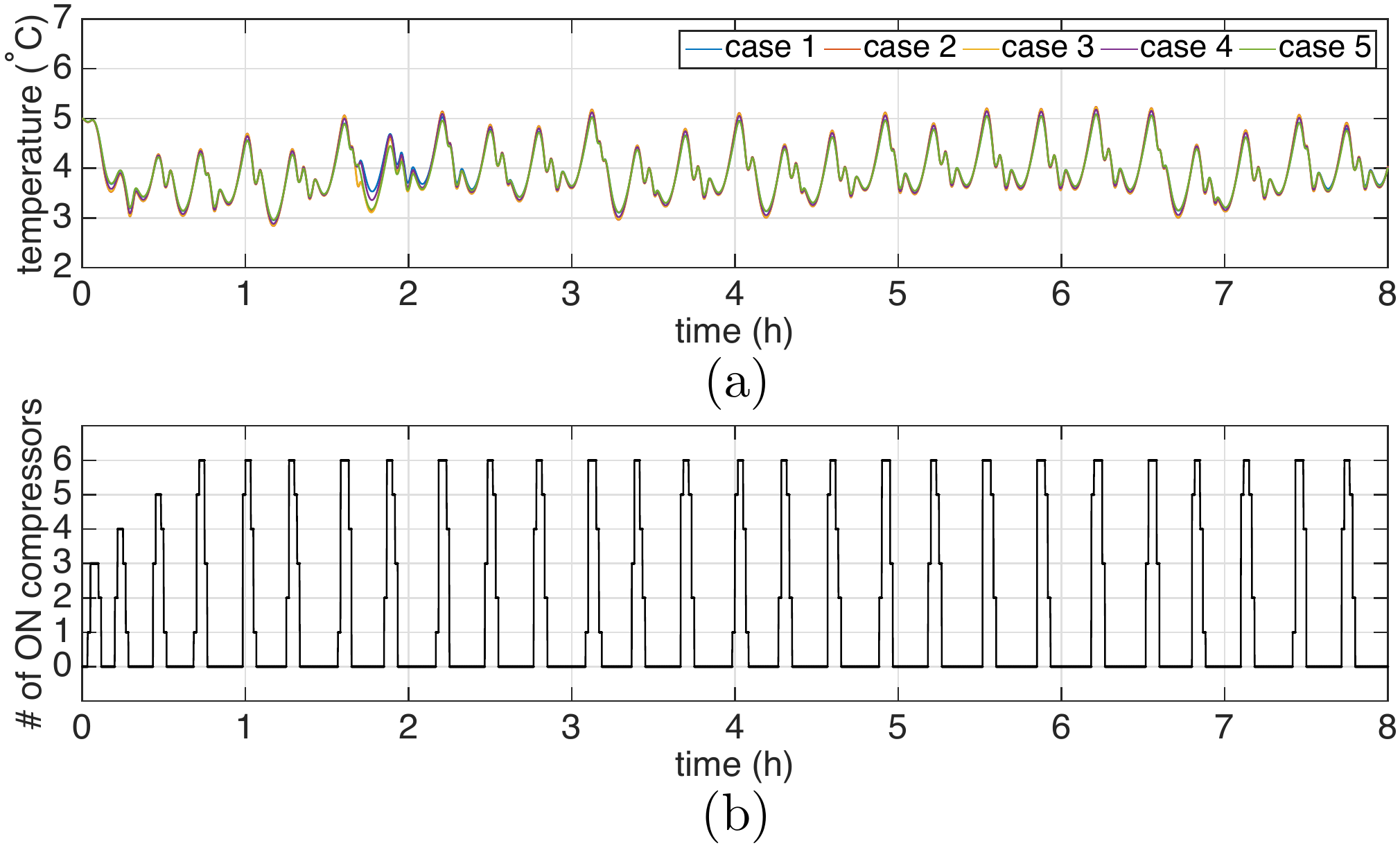}
  \caption{(a) The food temperatures in display cases 1--5, operated by the PI controller over 8 hours. The temperatures in all display cases are almost identical.
    (b) The number of ON compressors operated by the PI controller. Note that the profile presents frequent fluctuations.}
   \label{fig:PI}
\end{figure}

\section{Control via Online Combinatorial Optimization} \label{control}

\subsection{Conservative Compressor Control}\label{conservative}

We control the compressor rack to (approximately) maintain the suction pressure as the reference $\bar{P}_{\suc}$, i.e.,
\begin{equation}\nonumber
P_{\suc}(t) \approx \bar{P}_{\suc} \quad \forall t.
\end{equation}
 In other words, to make $\dot{P}_{\suc} \equiv 0$ in \eqref{p_dy},
 we set $u_c := (u_{c,1}, \cdots, u_{c,n_c})$ such that the refrigerant outflow from the suction manifold is equal to the inflow:
\begin{equation}\label{conservative}
\rho_{\suc} \sum_{i=1}^{n_c} k_{c} u_{c,i} \approx \sum_{i=1}^n f_i.
\end{equation}
In practice, we may not be able to exactly satisfy this equality because each $u_{c,i}$ is either $0$ or $1$. However, we assume that the compressor control action $u_{c}$ can be chosen to make  the difference between the outflow and the inflow negligible.  
This compressor control rule is suboptimal: it induces a conservative operation of the compressor rack that does not fully utilize the pressure bandwidth. 
However, this conservative control approach has a practical advantage: it does not create significant compressor switchings.
Therefore, it can potentially decelerate the mechanical wear of compressors. 
Under this compressor control rule, the total power consumption can be computed as
\begin{equation}\label{p_relation}
\begin{split}
p &= (h_{oc} - h_{ic}) \sum_{i=1}^n f_i \\
&=  \frac{(h_{oc} - h_{ic})m_{\mbox{\tiny ref}}^{\max}}{\Delta t}  \sum_{i=1}^n  u_i.
\end{split}
\end{equation}

\subsection{Bilevel Optimization Formulation}

We consider a receding-horizon online optimization approach to generate control signals for expansion valves and compressors.
Let $\{t_0, t_1, \cdots, t_k, t_{k+1}, \cdots\}$ be the time steps at which the control action is optimized.
For the sake of simplicity, we describe a one-step look-ahead optimization method; however, this approach can be easily extended to multiple-step look-ahead optimization (see Remark \ref{multi}).

\subsubsection{Inner problem for temperature management}
At time $t_k$, we control the expansion valves to minimize the following quadratic deviation
from the upper-bound $T_i^{\max}$, $i=1, \cdots, n$: 
\begin{equation}\nonumber
\begin{split}
J(\alpha) =\sum_{i=1}^n \int_{t_k}^{t_{k+1}}  (T_{\air, i} - {T}_{i}^{\max})_+^2 dt,
\end{split}
\end{equation}
where 
$(a)^2_{+} = a^2 \cdot \bold{1}_{\{a \geq 0\}}$,
assuming $T_{\air}$ is evolving with \eqref{food_dynamics} and \eqref{air_dynamics}.
Specifically, the expansion valve action at $t_k$ is generated as a solution to the following combinatorial optimization problem:
\begin{subequations}\label{opt_evap}
\begin{align}
\min_{\alpha \in \{0,1\}^n} \quad & 
J(\alpha)\\
\mbox{s.t.} \quad 
& \dot{x} = A x + Bu +  C, \quad x(t_k) = \bold{x}_{\meas} \label{sys}\\
& u(t) = \alpha, \quad t \in (t_k, t_{k+1}] \label{fix}\\
&\|\alpha \|_0 = \sum_{i=1}^n \alpha_i \leq K.  \label{power}
\end{align}
\end{subequations}
Here, $x := (T_{\food}, T_{\air})$ and \eqref{sys} gives a linear system representation of the dynamics \eqref{food_dynamics} and \eqref{air_dynamics}.
Note that $\bold{x}_{\meas}$ represents $(T_{\food}, T_{\air})$ measured at $t = t_k$.\footnote{If $T_{\food}$ is not directly measured, an observer needs to be employed to estimate the state.  Then, the control system uses the estimate $T_{\food}^{est}$ instead of its actual measurement as shown in Fig. \ref{fig:control}.} 
As specified in \eqref{fix}, the control action over $(t_k, t_{k+1}]$ is fixed as the solution $\alpha$.
The last constraint \eqref{power} essentially limits the power consumed by the refrigeration system as $K(h_{oc} - h_{ic})m_{\mbox{\tiny ref}}^{\max}/\Delta t$ due to \eqref{p_relation}. 
Therefore, the choice of $K$ is important to save energy: as $K$ decreases, the power consumption lessens.

\subsubsection{Outer problem for energy efficiency}

To generate an energy-saving control action, we minimize the number $K$ of open expansion valves while guaranteeing that the quadratic deviation $J(\alpha)$ from the upper-bound $T_i^{\max}$, $i=1, \cdots, n$ is bounded by the threshold $\Delta$. More precisely, we consider the following outer optimization problem:
\begin{equation} \label{outer}
\min \{ K \in \{0, \cdots, n\} \: | \: 
J(\alpha^{opt}(K)) \leq \Delta \},
\end{equation}
where $\alpha^{opt} (K)$ is a solution to the expansion valve optimization problem \eqref{opt_evap}.
Let $K^{opt}$ be a solution to this problem.
Then, $\alpha^{opt}(K^{opt})$ is the expansion valve control action that saves energy the most while limiting the violation of the food temperature upper-bound $T_i^{\max}$, $i=1, \cdots, n$.
This outer optimization problem can be easily solved by searching $K$ from $0$ in an increasing order. Once we find $\hat{K}$ such that $J(\alpha^{opt}(\hat{K})) \leq \Delta$, we terminate the search and obtain the solution as $K^{opt} := \hat{K}$.
In the following section, we will show that this procedure can be integrated into approximation algorithms for the inner optimization problem.

Then, as specified in \eqref{fix}, the controller chooses $u^{opt}(t) := \alpha^{opt}(K^{opt})$ for $t \in (t_k, t_{k+1}]$. Furthermore, it determines the compressor control signal $u_c^{opt}$ such that
 $\sum_{i=1}^{n_c} u_{c,i}^{opt} \approx m_{\mbox{\tiny ref}}^{\max}/ (\rho_{\suc} k_c \Delta t )$ using \eqref{conservative}.
 If $P_{\suc} (t) < \bar{P}_{\suc}$, the controller rounds 
 $m_{\mbox{\tiny ref}}^{\max}/ (\rho_{\suc} k_c \Delta t )$ to the next smaller integer and then determines the number of ON compressors as the integer.
If $P_{\suc} (t) \geq \bar{P}_{\suc}$, the controller rounds $m_{\mbox{\tiny ref}}^{\max}/ (\rho_{\suc} k_c \Delta t )$ to the nearest integer greater than or equal to it.
 The information flow in this control system is illustrated in Fig. \ref{fig:control}.

\begin{figure}
\begin{center}
  \includegraphics[width=2.3in]{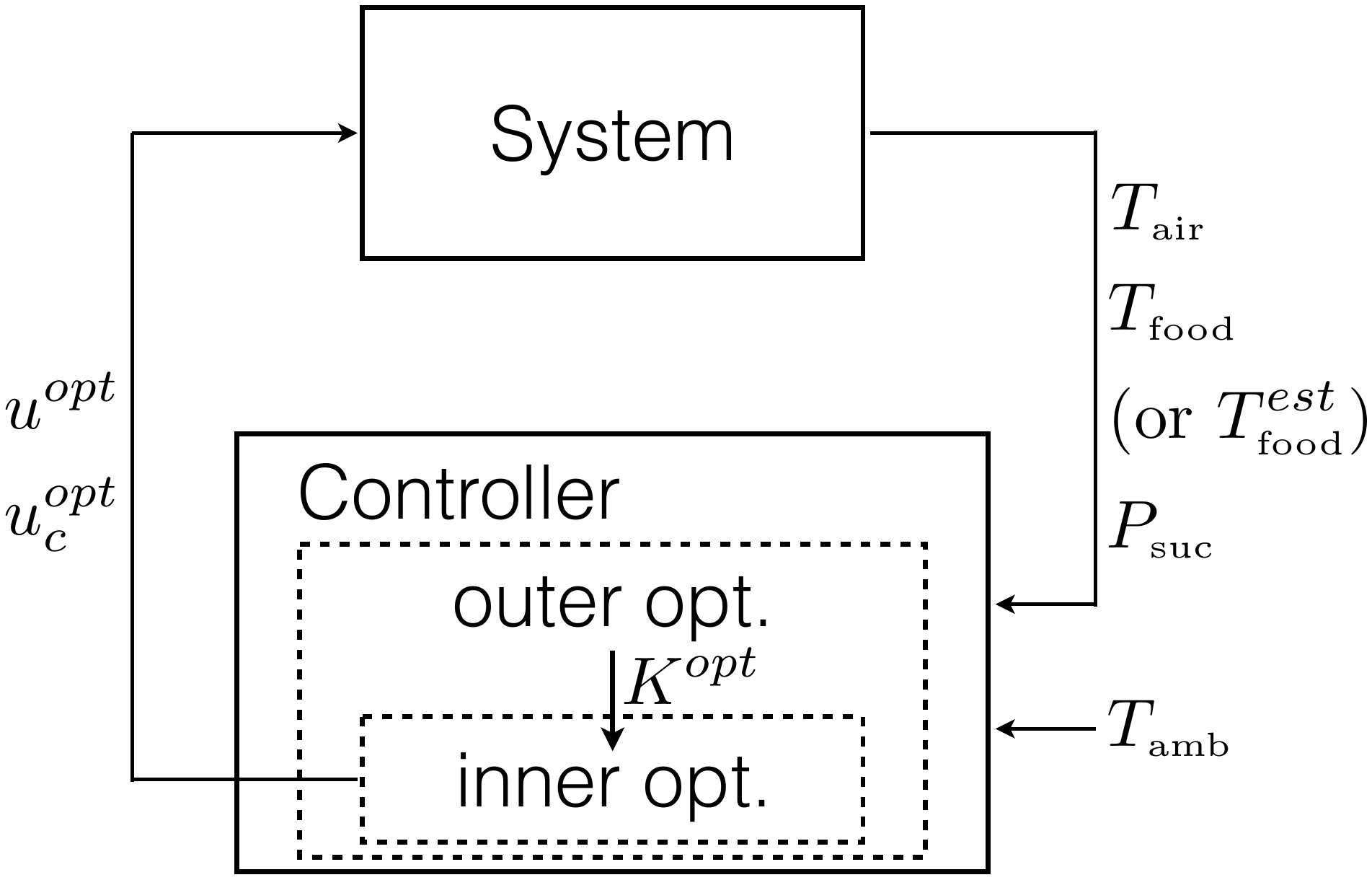}
  \caption{The proposed control system with the outer and inner optimization modules.}
   \label{fig:control}
   \end{center}
\end{figure}

\begin{remark}
Our objective function $J(\alpha)$ only takes into account the violation of temperature upper-bounds. 
This choice is motivated by the fact that the food temperature in each display case increases as we close more expansion valves, which is summarized in Proposition \ref{monotone}.
In other words, as we reduce the number $K$ of open valves in the outer optimization problem, 
 the possibility of violating temperature upper-bounds increases, while 
it is less likely to violate temperature lower-bounds.
This monotonicity property of food temperatures justifies our focus on temperature upper-bounds.
\end{remark}

\begin{proposition} \label{monotone}
Let $T_{\food, j}^\alpha$ and $T_{\air, j}^\alpha$ denote the food and air temperatures in display case $j$ when the control action $\alpha$ is applied.
Then, for any $\alpha, \beta \in \mathbb{R}^n$ such that
\begin{equation}\nonumber
\alpha_i \leq \beta_i, \quad i=1, \cdots, n,
\end{equation}
we have
\begin{equation}\nonumber
T_{\food, j}^\alpha \geq T_{\food, j}^\beta \mbox{ and }T_{\air, j}^\alpha \geq T_{\air, j}^\beta, \quad j=1, \cdots, n.
\end{equation}
\end{proposition}
\begin{proof}
The proof is contained in Appendix \ref{app_monotone}.
\end{proof}

\section{Approximation Algorithms}\label{app_alg}

We present two approximation methods for the inner optimization problem. One is based on linear approximation, and another utilizes submodularity.
These will give approximate solutions with guaranteed suboptimality bounds.
We further show that, by simply modifying these approximation algorithms for the inner optimization problem, we can obtain a near-optimal solution to the outer optimization problem.

\subsection{Linear Approximation-Based Optimization}

We first consider a linear approximation-based approach to the inner combinatorial optimization problem \eqref{opt_evap}.
It is convenient to work with the following value function:
\begin{equation}\label{opt_submodular}
{V}(\alpha) = J(0)-J(\alpha).
\end{equation}
The value ${V}(\alpha)$ represents the reduction in the quadratic deviation from from the upper-bound $T_i^{\max}$, $i=1, \cdots, n$, when  expansion valve $j$ is chosen to be open only for $j$ such that $\alpha_j = 1$.
Note that this value function is normalized such that $V(0) = 0$.
The Taylor expansion of $V$ at $0$ gives $V(\alpha) = DV(0)^\top \alpha + O(\alpha^2)$ assuming the derivative $DV$ is well-defined. This motivates us to consider the following first-order approximation of the expansion valve optimization problem \eqref{opt_evap}:
\begin{equation}\label{l_0_app}
\begin{split}
\max_{\alpha \in \{0,1\}^n} \quad & 
DV(0)^\top \alpha \\
&\|\alpha \|_0 \leq K. 
\end{split}
\end{equation}
The $i$th entry $[DV(0)]_i$ of the derivative represents the marginal benefit of opening expansion valve $i$.
Therefore, the approximate problem \eqref{l_0_app} can be interpreted as maximizing the marginal benefit of valve operation while guaranteeing that the number of open valves is less than or equal to $K$.  A detailed method to define and compute the derivative can be found in \cite{Yang2016}.
Computing the derivative should also take into account the dependency of the state $x$ on the binary decision variable $\alpha$. For example, an adjoint-based approach can be used to handle this dependency \cite{Kokotovic1967}.

The first advantage of the proposed approximation approach is 
that it gives an approximate solution with a provable suboptimality bound.
The bound is \emph{a posteriori}, which does not require the globally optimal solution $\alpha^{opt}$ but the solution $\alpha^{\star}$ of \eqref{l_0_app}.
\begin{theorem}[\cite{Yang2016}] \label{thm:lin}
Let $\alpha^\star$ be a solution to the approximate problem \eqref{l_0_app}. If $DV(0)^\top \alpha^\star \neq 0$, then
the following suboptimality bound holds:
\begin{equation}\nonumber
\rho V(\alpha^{opt}) \leq V(\alpha^\star),
\end{equation}
where
\begin{equation}\nonumber
\rho = \frac{V(\alpha^\star)}{DV(0)^\top \alpha^\star} \leq 1.
\end{equation}
If $DV(0)^\top \alpha^\star = 0$, then $V(\alpha^{opt}) = V(0) = 0$, i.e., $0$ is an optimal solution.
\end{theorem}
Its proof is contained in Appendix \ref{thm_lin}.
This theorem suggests that the approximate solution's performance is greater than $(\rho \times 100)$\% of the globally optimal solution's performance. 

The second advantage of the proposed method is that
it yields an efficient algorithm to solve the approximate problem \eqref{l_0_app}.
Specifically, we design a very simple algorithm based on the ordering  of the entries of $DV(0)$.
Let $\bold{d} (\cdot)$ denote the map from $\{1, \cdots, n\}$ to $\{1, \cdots, n\}$ such that 
\begin{equation} \label{des}
[DV(0)]_{\bold{d}(i)} \geq [DV(0)]_{\bold{d}(j)}
\end{equation}
 for any $i, j \in \{1, \cdots, n\}$ such that $i \leq j$. Such a map can be constructed using a sorting algorithm with $O(n \log n)$ complexity. 
Such a map may not be unique.
We  let $\alpha_{\bold{d}(i)} = 1$ for $i=1, \cdots, K$ if $[DV(0)]_{\bold{d}(i)} > 0$.
A more detailed algorithm to solve this problem  is presented in Algorithm~\ref{algorithm:l_0}.
Note that it is a single-pass algorithm, i.e., does not require multiple iterations.
Furthermore, Lines 9--12 can be parallelized.

\begin{algorithm}

\textbf{Initialization:}\\
$\alpha \leftarrow 0$;

\vspace{0.05in}
\textbf{Construction of $\bold{d}$:}\\
Compute $DV(0)$;\\
Sort the entries of $DV(0)$ in descending order;\\
Construct $\bold{d}: \{1, \cdots, n\} \to \{1, \cdots, n\}$ satisfying \eqref{des};

\vspace{0.05in}
\textbf{Solution of \eqref{l_0_app}:}\\
\While{$[DV(0)]_{\bold{d}(i)} > 0$ and $i \leq K$}{
$\alpha_{\bold{d}(i)} \leftarrow 1$;\\
$i \leftarrow i+1$;
}
\caption{Algorithm for the approximate problem \eqref{l_0_app}
}
\label{algorithm:l_0}
\end{algorithm}

\begin{remark}\label{multi}
The proposed linear approximation method is applicable to multi-period optimization problems, in which the objective  is given by $J(\alpha) := \sum_{k=1}^{N_{period}} J_k (\alpha^k)$ and
the control variable is time-varying, i.e., $\alpha = (\alpha^1, \cdots, \alpha^{N_{period}}) \in \mathbb{R}^{n\times N_{period}}$. In such a case, we  compute the derivative $DV_k$ of $V_k(\alpha^k) := J_k(0) - J_k(\alpha^k)$
 for each $k$. 
The objective function can be approximated as $\sum_{k=1}^{N_{period}} DV_k(0)^\top \alpha^k$, which is still linear in $\alpha$. Therefore, we can use the proposed algorithm.
\end{remark}

\subsection{Submodular Optimization}

The second approach gives another approximate solution of the expansion valve optimization problem \eqref{opt_evap} with a  suboptimality bound.  This solution is generally different from the solution obtained by the first approach.
Let $\Omega := \{1, \cdots, n\}$ be the set of expansion valves to be controlled.
We define a set function, $\mathcal{V}: 2^\Omega \to \mathbb{R}$, as
\begin{equation}\nonumber
\mathcal{V}(X) = V(\mathbb{I}(X)),
\end{equation}
where the value function $V$ is defined as \eqref{opt_submodular} and
$\mathbb{I}(X) := (\mathbb{I}_1(X), \cdots, \mathbb{I}_n(X)) \in \{0,1\}^n$ is the indicator vector of the set $X$ such that
$\mathbb{I}_i(X) := 
0$ if $i \notin X$ and
$\mathbb{I}_i(X) := 
1$ if $i \in X$. 
In other words, $\mathcal{V}$ is a set function representation of $V$. 
The expansion valve optimization problem \eqref{opt_evap} is equivalent to selecting the set $X \subseteq \Omega$ such that $|X| \leq K$ 
to maximize the value function $\mathcal{V}(X)$, i.e.,
\begin{equation}\label{opt_submodular}
\begin{split}
\max_{X \in 2^\Omega} \quad &\mathcal{V}(X) \\
\mbox{s.t.}\quad &|X|\leq K.
\end{split}
\end{equation}
We observe that the value function $\mathcal{V}$ has a useful structure, which is called the \emph{submodularity}.
It represents a \emph{diminishing return} property such that opening an expansion valve when a smaller set of valves is opened gives a greater marginal benefit than opening it when a larger set of valves is already opened.

\begin{theorem}\label{submodular}
The set function $\mathcal{V} : 2^\Omega \to \mathbb{R}$ is submodular, i.e.,
for any $X \subseteq Y \subseteq \Omega$ and any $\bold{a} \in \Omega \setminus Y$,
\begin{equation}\nonumber
\mathcal{V}(X \cup \{\bold{a}\}) - \mathcal{V}(X) \geq \mathcal{V}(Y \cup \{\bold{a}\}) - \mathcal{V}(Y).
\end{equation}
Furthermore, it is monotone, i.e., for any $X \subseteq Y \subseteq \Omega$
\begin{equation}\nonumber
\mathcal{V}(X) \leq \mathcal{V} (Y).
\end{equation}
\end{theorem}
\begin{proof}
See Appendix \ref{app_submodular}.
\end{proof}

The submodularity of $\mathcal{V}$ guarantees that Algorithm \ref{algorithm:submodular}, which is a greedy algorithm, provides an $(1-\frac{1}{e})$-optimal solution. In other words, the approximate solution's performance is greater than $(1-\frac{1}{e}) \approx 63\%$ of the oracle's performance.
In our case study, the actual submodularity is $98.9\%$, which is significantly greater than this theoretical bound.

\begin{algorithm}

\textbf{Initialization:}\\
$X \leftarrow \emptyset$;

 \vspace{0.05in}
\textbf{Greedy algorithm:}\\
\While{$i \leq K$}{
$\bold{a}^* \in \arg \max_{\bold{a} \in \Omega \setminus X} \mathcal{V}( X \cup \{\bold{a}\})$;\\
$X \leftarrow  X \cup \{\bold{a}^*\}$;\\
$i \leftarrow i+1$;
}
\caption{Greedy algorithm for \eqref{opt_submodular}}
\label{algorithm:submodular}
\end{algorithm}

\begin{theorem}[\cite{Nemhauser1978}]
Algorithm \ref{algorithm:submodular} 
is a $\left(1 - \frac{1}{e}\right)$-approximation algorithm.  In other words, if we let $X^\star$ be the solution obtained by this greedy algorithm, then the following suboptimality bound holds:
\begin{equation}\nonumber
\left(1 - \frac{1}{e}\right) \mathcal{V}(X^{opt}) \leq \mathcal{V}(X^\star),
\end{equation}
where $X^{opt}$ is an optimal solution to \eqref{opt_submodular}.
\end{theorem}

Lines 5--9 of Algorithm \ref{algorithm:submodular} makes a locally optimal choice at each iteration. Therefore, it significantly reduces the search space, i.e., it does not search over all possible combinations of open expansion valves.  
When the expansion valve optimization problem \eqref{opt_evap} is extended to multi-stage optimization, 
a greedy algorithm achieves the same suboptimality bound 
using
adaptive (or string) submodularity and monotonicity  of the value function \cite{Golovin2011, Alaei2010, Liu2014}.

\subsection{Modified Algorithms for the Outer Optimization Problem}

We now modify the two approximation algorithms for the inner problem \eqref{opt_evap} to solve the full bilevel optimization problem.
In both Algorithms \ref{algorithm:l_0} and \ref{algorithm:submodular},
the expansion valve chosen to be open at iteration $i$ (line 10 of Algorithm \ref{algorithm:l_0} and line 7 of Algorithm \ref{algorithm:submodular}) is \emph{independent of the selections at later iterations}.
This independency  plays an essential role in incorporating the outer optimization problem into the algorithms.
To be more precise, we compare the cases of $K = l$ and $K = l+1$.  Let $\alpha^{l}$ and $\alpha^{l + 1}$ be the solutions in the two cases obtained by Algorithm \ref{algorithm:l_0}.
Since the expansion valve selected to be open at  iteration $l+1$ do not affect the choices at earlier iterations,
we have $\alpha^{l}_{\bold{d}(i)} = \alpha^{l+1}_{\bold{d}(i)}$ for $i=1, \cdots, l$.
Therefore, we do not have to re-solve the entire inner optimization problem for $K = l+1$ if we already have the solution for $K = l$; it suffices to run one more iteration for $i = l+1$ to obtain $\alpha^{l+1}_{\bold{d}(l +1)}$.
This observation allows us to simply modify lines 8--12 in Algorithm \ref{algorithm:l_0} as Algorithm \ref{algorithm:l_0_modified}.
\begin{algorithm}
\While{$[DV(0)]_{\bold{d}(i)} > 0$ and $J(\alpha) > \Delta$}{
$\alpha_{\bold{d}(i)} \leftarrow 1$;\\
$i \leftarrow i+1$;
}
\caption{Modified version of Algorithm 1 for the outer optimization problem \eqref{outer}
}
\label{algorithm:l_0_modified}
\end{algorithm}
As we can see in line 2, we select expansion valves to be open until when the temperature upper-bound violation $J(\alpha)$ is less than or equal to the threshold $\Delta$.
Similarly, we modify lines 5--9 of Algorithm \ref{algorithm:submodular} as Algorithm \ref{algorithm:submodular_modified} to solve the outer problem.

\begin{algorithm}
\While{$J(\mathbb{I}(X)) > \Delta$}{
$\bold{a}^* \in \arg \max_{\bold{a} \in \Omega \setminus X} \mathcal{V}( X \cup \{\bold{a}\})$;\\
$X \leftarrow  X \cup \{\bold{a}^*\}$;\\
$i \leftarrow i+1$;
}
\caption{Modified version of Algorithm 2 for the outer optimization problem \eqref{outer}}
\label{algorithm:submodular_modified}
\end{algorithm}

\subsection{Scalability}

We now compare the complexity of Algorithms \ref{algorithm:l_0_modified} and \ref{algorithm:submodular_modified}.
Algorithm \ref{algorithm:l_0_modified}, which is based on a linear approximation, is \emph{single pass} in the sense that, after computing the derivative and ordering its entries only once, we can obtain the solution.
Calculating the derivative requires $O(n^2N_T)$, where $N_T$ is the number of time points in the time interval $[t_k, t_{k+1}]$ used to integrate the dynamical system \cite{Yang2016}, if a first-order scheme is employed.
Therefore, the total complexity including the sorting step is $O(n^2 N_T) + O(n \log n)$. 
On the other hand, Algorithm \ref{algorithm:submodular_modified}, which is a greedy algorithm, chooses a locally optimal solution at each stage.
In other words, this iterative greedy choice approach requires one to find an entry that maximizes the increment in the current payoff at every stage.
Its complexity is $O(n^3 N_T)$.
Therefore, Algorithm \ref{algorithm:l_0_modified} is computationally more efficient as the number $n$ of display cases grows.
Note, however, that Algorithm \ref{algorithm:submodular_modified} is also scalable because
its complexity is cubic in $n$ and linear in $N_T$.

\section{Case Studies} \label{case}

In this section, we examine the performance of the proposed online optimization-based controllers.
For fair comparisons with the traditional controller, 
we use the parameter data reported in Section \ref{convention} with a refrigerator unit that has 10 display cases.
Fig. \ref{fig:temp} and \ref{fig:compressor} illustrates the simulation results of 
 the proposed controllers.

\subsection{Energy Efficiency}

As opposed to the synchronized food temperature profiles controlled by the traditional method (see Fig. \ref{fig:PI} (a)),
the proposed controllers induce alternating patterns of the temperatures as shown in Fig. \ref{fig:temp}.
Such patterns result from the explicit consideration of heat transfers between neighboring display cases in the optimization module through the constraint \eqref{sys}, which represents the interconnected temperature dynamics.
Using the spatial heat transfers, the proposed controllers do not turn ON or OFF all the expansion valves at the same time. 
Instead, they predict the temperature evolution for a short period and selects the valves to turn ON that are effective to minimize the deviation from the desirable temperature range during the period.
As a result, the ON/OFF operation of expansion valves is desynchronized, unlike in the case of the traditional controller.
This desynchronization maintains the temperatures near the upper-bound $T^{\max}$ reducing temperature fluctuations.
Therefore, it intuitively improves energy efficiency.
As summarized in Table \ref{efficiency},
the proposed controllers save $7.5$--$8\%$ of energy.

\begin{center}
    \captionof{table}{Energy savings by the proposed controllers} \label{efficiency}
  \begin{tabular}{   l || p{0.6in} | p{0.6in}| p{0.6in}}
 & PI & linear & submodular  \\ \hline
average kW  &11.24  & 10.34 & 10.40 \\
energy saving & -- & 8.0\% & 7.5\%\\
suboptimality & 90.1\%  & 99.5\%   & 98.9\%
  \end{tabular}
\end{center}

Note that the outer optimization module minimizes the total energy consumption while the inner optimization module is responsible for maintaining the temperature profiles in a desirable range.
When the bilevel combinatorial problem is exactly solved for all time, 
the average power consumption is $10.29$kW.
Therefore, the two proposed controllers' performances are $99.5\%$ and $98.9\%$ of the optimal controller although their theoretical suboptimality bounds are $39\%$ and $63\%$.

\begin{figure}[tp]
\centering
  \includegraphics[width=3.5in]{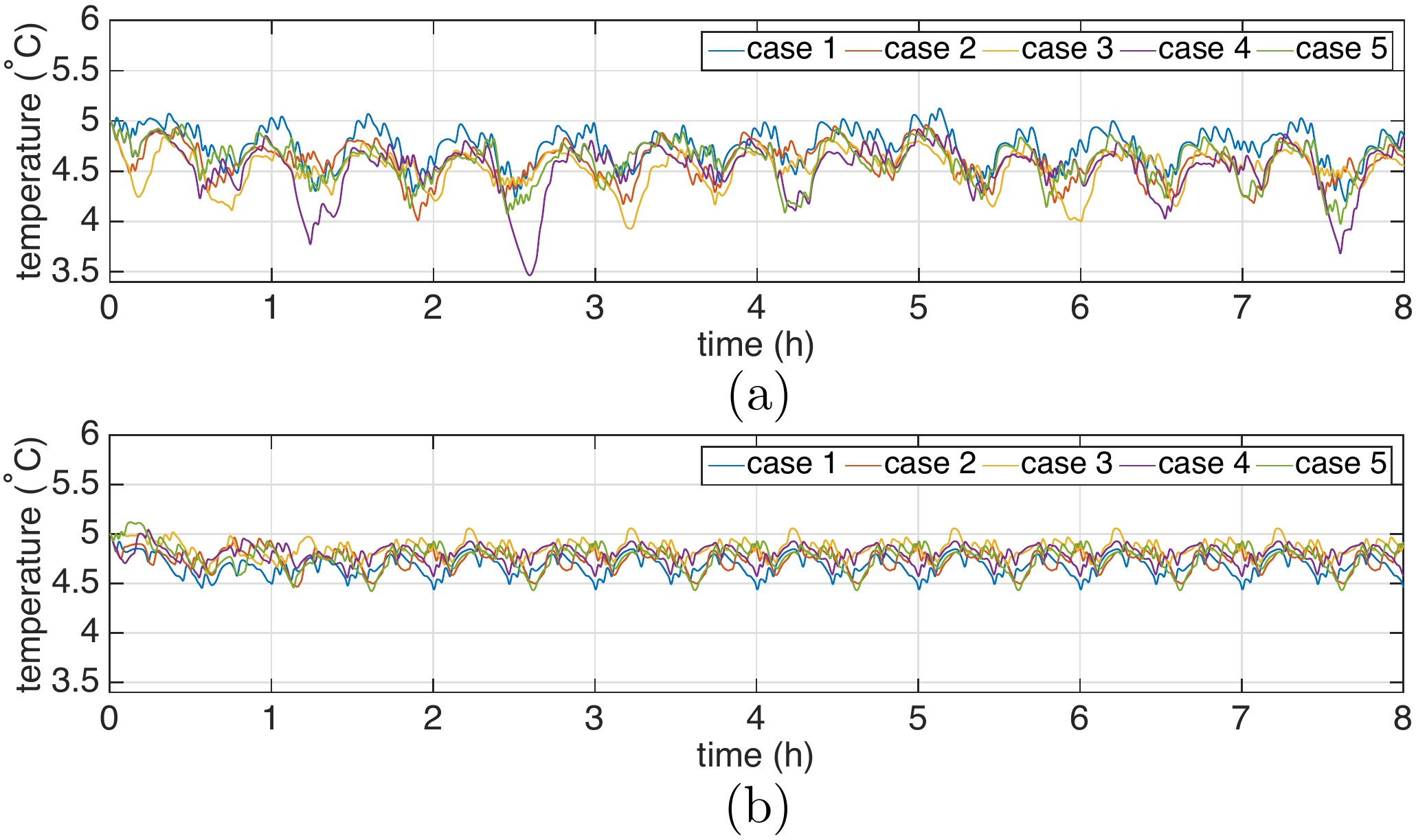}
  \caption{The food temperatures (in 5 display cases out of 10) controlled by (a) the linear approximation-based algorithm (Algorithm \ref{algorithm:l_0_modified}), and (b) the submodular optimization algorithm (Algorithm \ref{algorithm:submodular_modified}).   }\label{fig:temp}
\end{figure}

\begin{figure}[tp]
\centering
  \includegraphics[width=3.5in]{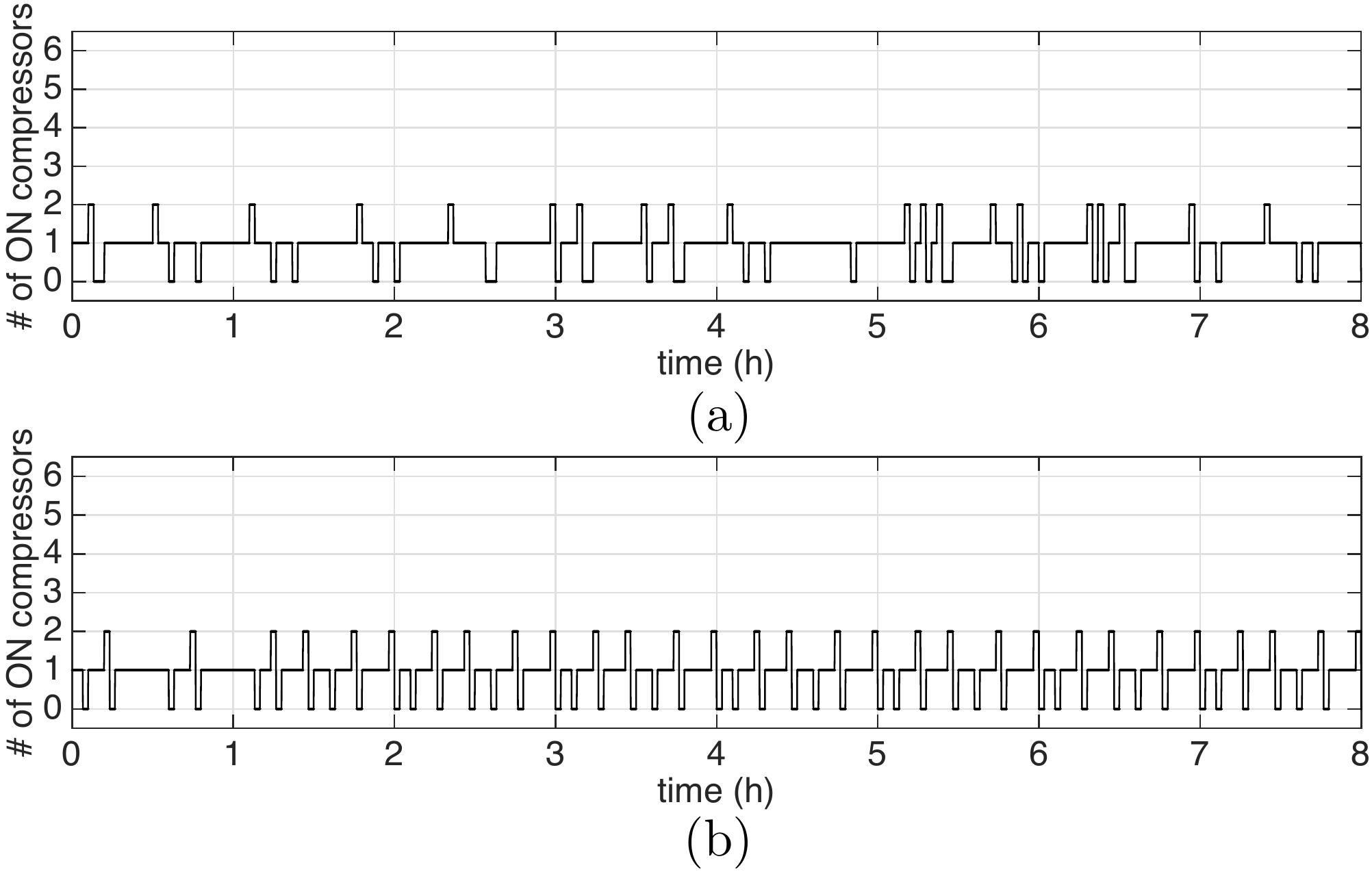}
  \caption{The number of ON compressors controlled by (a) the linear approximation-based algorithm (Algorithm \ref{algorithm:l_0_modified}), and (b) the submodular optimization algorithm (Algorithm \ref{algorithm:submodular_modified}). }\label{fig:compressor}
\end{figure}

\subsection{Reduced Compressor Switching}

Another advantage of  the proposed controllers is
the considerable reduction on the number of compressor switching instances.
By desynchronizing the switching instances of expansion valves in the inner optimization module, the proposed controllers significantly reduce the variation of suction pressure.
Our conservative compressor control approach presented in Section \ref{conservative} also helps to minimize the deviation of the suction pressure from its reference.
As a result, the controllers significantly reduce the fluctuations on the number of ON compressors as shown in Fig. \ref{fig:compressor}.
First, the maximum number of ON compressors is decreased from six to two.  This reduction suggests that a mechanically more compact compressor or a smaller number of compressors in the rack may be enough if the proposed controllers are adopted.
Second, the proposed controllers reduce
the number of compressor switching instances by 54.0--71.6\% as summarized in Table \ref{table:switching}.
These infrequent compressor operation strategies are beneficial to decelerate the mechanical degradation of compressors.

\begin{center}
\captionof{table}{Compressor switching reductions by the proposed controllers} \label{table:switching}
  \begin{tabular}{   l || p{0.6in} | p{0.6in}| p{0.6in}}
 & PI & linear & submodular  \\ \hline
\# of switchings  & 324  & 92 & 149 \\
reduction & -- & 71.6\% & 54.0\%
  \end{tabular}
\end{center}

\subsection{Comparisons of the Two Proposed Controllers}

Fig. \ref{fig:temp} illustrates that the submodular optimization-based controller maintains the temperatures in a narrower range than the linear approximation-based controller.
This feature is owing to the fact that the greedy algorithm used in the submodular optimization-based method avoids violating the temperature upper bound in a locally optimal way.
However, to keep the temperatures in a narrower range
this approach requires a faster adjustment of suction pressure than the linear approximation-based method.
As a result, the proposed compressor controller performs a more frequent switching of compressors when the submodular optimization-based method is adopted (see Fig. \ref{fig:compressor}  and Table \ref{table:switching}).
Furthermore, this frequent compressor switching induces an inefficient use of the compressor rack and therefore turns ON more compressors on average (in time).
Therefore, the submodular optimization-based controller consumes slightly more energy than the linear approximation-based controller as reported in Table \ref{efficiency}.

\subsection{Automated Demand Response under Real-Time Pricing}\label{dr}

\emph{Real-time pricing} of electricity refers to passing wholesale prices through to end users. At least in theory, it is shown to improve electricity market efficiency among other benefits~\cite{Borenstein2005}.\footnote{However, real-time prices fail to capture the economic value of responsive loads in general \cite{Sioshansi2009}.}
However, consumer should bear the risk of receiving high energy bills if consumers do not appropriately react to fluctuating wholesale prices.
Such a risk transfer to end-users under real-time pricing can be reduced by automated demand response (ADR) \cite{Piette2008} and can also be limited by contracts for ADR \cite{Yang2015acc}.
In this subsection, we demonstrate the utility of our method as a control tool for refrigeration ADR under real-time pricing. 
In particular, we consider the scenario of energy arbitrage: supermarket refrigeration systems automatically consume less energy when the real-time price is high and consume more when it is low.
Note that real-time prices are often difficult to predict and hence ADR must be performed in an online fashion by appropriately reacting to  fluctuating prices. The online optimization feature of our controllers allows them to adjust energy consumption in response to real-time prices (by changing the number $K$ of ON expansion valves in real time). 
Fig. \ref{fig:arbitrage} (a) shows the real-time wholesale electricity price  at the Austin node in the Electric Reliability Council of Texas (ERCOT) on July 3, 2013 \cite{ERCOT}. 
We use a simple thresholding law for choosing $K$: if the electricity price is greater than \$0.1/kWh, the controller allows up to 70\% of expansion valves to turn ON; otherwise, it operates as usual.
As summarized in Table \ref{table:arbitrage}, the proposed controllers save 14.3--15.0\% of energy cost compared to a standard PI controller.
In addition, the temperature deviations from $T^{\max}$ are less than $0.5^\circ$C as shown in Fig. \ref{fig:arbitrage} (b), because right after the reduction in energy consumption the controller encourages enough cooling to recover the desired temperature levels.\footnote{Such post-cooling is mostly feasible  due to the mean-reverting behavior of electricity prices which discourages sustained price peaks \cite{Cartea2005}. We can also perform pre-cooling if prices are predictable or their distributional information is available. Such pre-cooling will increase the economic benefit of the proposed method.  
}
Further applications of the proposed algorithms to ADR aggregating a large number of supermarket refrigerator units can be found in \cite{Yang2015}.

\begin{figure}[tp]
\centering
  \includegraphics[width=3.5in]{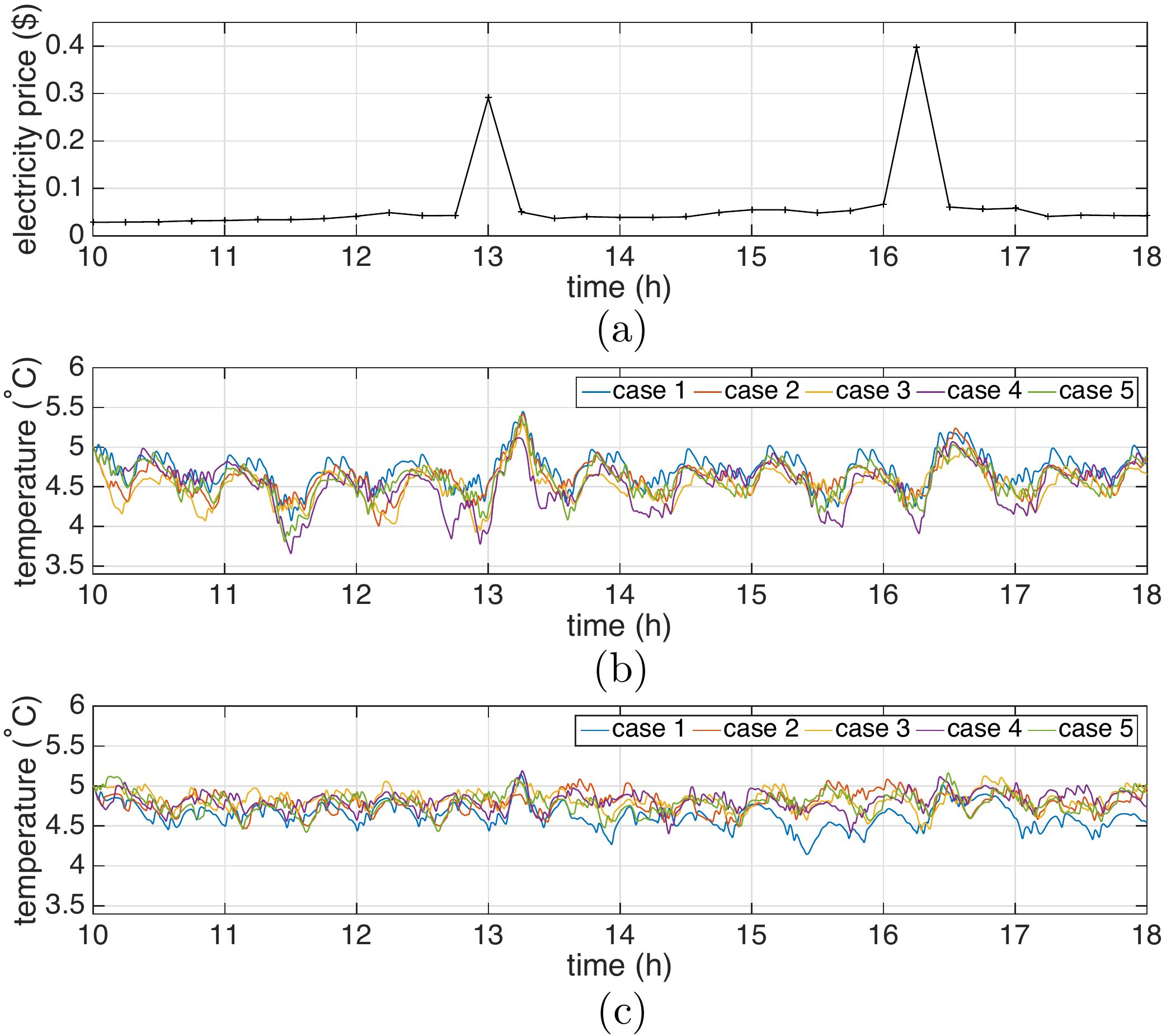}
  \caption{(a) The real-time price data;
  The food temperatures (in five display cases out of ten) controlled by (b) the linear approximation-based algorithm (Algorithm \ref{algorithm:l_0_modified}), and (c) the submodular optimization algorithm (Algorithm \ref{algorithm:submodular_modified}). }\label{fig:arbitrage}
\end{figure}

\begin{center}
\captionof{table}{Operation costs per refrigerator under real-time pricing from 10am to 6pm} \label{table:arbitrage}
  \begin{tabular}{   l || p{0.6in} | p{0.6in}| p{0.6in}}
 & PI & linear & submodular  \\ \hline
cost  & \$5.67  & \$4.82 & \$4.86 \\
cost saving & -- & 15.0\% & 14.3\%
  \end{tabular}
\end{center}

\section{Conclusions}

The proposed controller explicitly takes into account the switched and interconnected dynamics,
and is therefore  is suitable for multi-case supermarket refrigeration systems.
However, it has to solve a bilevel combinatorial optimization problem associated with switched interconnected systems in near real time, which is a challenging task. 
To overcome this difficulty, we proposed two polynomial-time approximation algorithms that are based on the structural properties of this optimization problem. These algorithms can also be adopted as a tool for solving combinatorial optimization problems arising in existing MPC-based methods.
We demonstrated the performance of the proposed controllers through case  studies using a benchmark refrigeration system model and found that $(i)$ they improve energy efficiency by 7.5--8\%,  $(ii)$ they reduce the number of compressor switchings by 54--71.6\%, and $(iii)$ they save 14.3--15\% of operation cost under a demand response scenario. In addition to conventional usages, the scalability of the proposed algorithms can contribute to 
an emerging methodology of controlling a large number of refrigerator units through cloud computing.

\appendix

\section{Proof of Proposition \ref{monotone}} \label{app_monotone}

\begin{proof}
We use the linear system representation \eqref{sys} of the food and air temperature dynamics (equations \eqref{food_dynamics} and \eqref{air_dynamics}).
We first notice that
\begin{equation}\nonumber
A_{i,j}\geq0 \quad \forall i \neq j,
\end{equation}
where $A_{i,j}$ represents the $(i,j)$th entry of the matrix $A$.
Furthermore, $k_{\air-\evap} T_{\evap} \leq 0$ due to the non-positive evaporator temperature. Hence, we have
\begin{equation}\nonumber
B_{i,j} \leq 0 \quad \forall i,j.
\end{equation}
Using Proposition III.2 in \cite{Angeli2003}, we conclude that the system \eqref{sys} is \emph{input-monotone} such that 
for any $\alpha, \beta \in \mathbb{R}^n$ with $\alpha_i \leq \beta_i$, $i=1, \cdots, n$,
\begin{equation}\nonumber
x_i^\alpha  \geq x_i^\beta, \quad i=1, \cdots, n,
\end{equation}
where $x^\alpha$ denotes the solution of the system \eqref{sys} when its input is chosen as $\alpha$. 
\end{proof}

\section{Proof of Theorem \ref{thm:lin}}\label{thm_lin}
\begin{proof}
In \eqref{sys}, we notice that 
\begin{equation}\nonumber
x(t) = e^{A (t - t_k)} \bold{x}_{\meas} + \int_{t_k}^t 
e^{A (t - s)} B \alpha ds,
\end{equation}
which implies that $x(t)$ is linear in $\alpha$. Therefore, $V$ is concave with respect to $\alpha$ in a continuously relaxed space, $\mathbb{R}^n$. Then, the result follows from Theorem 2 in \cite{Yang2016}.
\end{proof}

\section{Proof of Theorem \ref{submodular}}\label{app_submodular}
\begin{proof}
Let $T_{\food, i}^X$ denote the temperature of the food product in display case $i$ given that the expansion valves in $X$ are open.
Due to the linearity of the system dynamics \eqref{food_dynamics} and \eqref{air_dynamics}, $T_{\food, i}^X$ is modular, i.e.,
\begin{equation} \nonumber
T_{\food, i}^X = \sum_{\bold{a} \in X} T_{\food, i}^{\{\bold{a}\}}.
\end{equation}
Therefore, for any $X \subseteq Y \subseteq \Omega$  and any $\bold{a} \in \Omega \setminus Y$
\begin{equation}\nonumber
T_{\food, i}^{X \cup \{\bold{a}\}} - T_{\food, i}^X= 
T_{\food, i}^{Y \cup \{\bold{a}\}} - T_{\food, i}^Y.
\end{equation}
Furthermore, Proposition \ref{monotone} yields the following monotonicity result: for any $X \subseteq Y \subseteq \Omega$
\begin{equation}\nonumber
T_{\food, i}^X \geq T_{\food, i}^Y,
\end{equation}
i.e., as we open more expansion valves, the food temperature decreases.
Lastly, the concavity of $\mathcal{V}(X) = \mathcal{V}(\emptyset) - 
\sum_{i=1}^n \int_{t_k}^{t_{k+1}}  (T_{\food, i}^X - {T}_{i}^{\max})_+^2 dt$ in $T_{\food, i}^X$ implies that  $X \subseteq Y \subseteq \Omega$  and any $\bold{a} \in \Omega \setminus Y$
\begin{equation}\nonumber
\mathcal{V}(X \cup \{\bold{a}\}) - \mathcal{V}(X) \geq
\mathcal{V}(Y \cup \{\bold{a}\}) - \mathcal{V}(Y). 
\end{equation}
Therefore, $\mathcal{V}$ is submodular. Its monotonicity follows from Proposition \ref{monotone}.
\end{proof}

\bibliographystyle{plain}

\bibliography{reference}

\begin{thebibliography}{10}

\bibitem{ES2008}
{\em {ENERGY STAR} Building Upgrade Manual Chapter 11: Supermarkets and Grocery
  Stores}, 2008.

\bibitem{Alaei2010}
S.~Alaei and A.~Malekian.
\newblock Maximizing sequence-submodular functions and its application to
  online advertising.
\newblock {\em arXiv:1009.4153 [cs.DM]}, 2010.

\bibitem{Angeli2003}
D.~Angeli and E.~D. Sontag.
\newblock Monotone control systems.
\newblock {\em IEEE Transactions on Automatic Control}, 48(10):1684--1698,
  2003.

\bibitem{Bemporad1999}
A.~Bemporad and M.~Morari.
\newblock Control of systems integrating logic, dynamics, and constraints.
\newblock {\em Automatica}, 35(3):407--427, 1999.

\bibitem{Borenstein2005}
S.~Borenstein and S.~P. Holland.
\newblock On the efficiency of competitive electricity markets with
  time-invariant retail prices.
\newblock {\em RAND Journal of Economics}, 36(3):469--493, 2005.

\bibitem{Cartea2005}
A.~Cartea and M.~G. Figueroa.
\newblock Pricing in electricity markets: a mean reverting jump diffusion model
  with seasonality.
\newblock {\em Applied Mathematical Finance}, 12(4):313--335, 2005.

\bibitem{ERCOT}
{Electric Reliability Council of Texas}.
\newblock Real-time prices reports.

\bibitem{Golovin2011}
D.~Golovin and A.~Krause.
\newblock Adaptive submodularity: Theory and applications in active learning
  and stochastic optimization.
\newblock {\em Journal of Artificial Intelligence Research}, 42:427--486, 2011.

\bibitem{Graziano2014}
M.~Graziano and M.~Pritoni.
\newblock Gloudfridge: A cloud-based control system for commercial
  refrigeration systems.
\newblock Technical report, ACEEE Summer Study on Energy Efficiency in
  Buildings, 2014.

\bibitem{Hovgaard2013}
T.~G. Hovgaard, S.~Boyd, L.~F.~S. Larsen, and J.~B. J{\o}rgensen.
\newblock Nonconvex model predictive control for commercial refrigeration.
\newblock {\em International Journal of Control}, 86(8):1349--1366, 2013.

\bibitem{Kokotovic1967}
P.~Kokotovi\'{c} and J.~Heller.
\newblock Direct and adjoint sensitivity equations for parameter optimization.
\newblock {\em IEEE Transactions on Automatic Control}, 12(5):609--610, 1967.

\bibitem{Larsen2005}
L.~F.~S. Larsen, T.~Geyer, and M.~Morari.
\newblock Hybrid model predictive control in supermarket refrigeration systems.
\newblock In {\em Proceedings of 16th IFAC World Congress}, 2005.

\bibitem{Larsen2007}
L.~F.~S. Larsen, R.~Izadi-Zamanabadi, and R.~Wisniewski.
\newblock Supermarket refrigeration system - benchmark for hybrid system
  control.
\newblock In {\em Proceedings of 2007 European Control Conference}, 2007.

\bibitem{Li2010}
B.~Li and A.~G. Alleyne.
\newblock A dynamic model of a vapor compression cycle with shut-down and
  start-up operations.
\newblock {\em International Journal of Refrigeration}, 33(3):538--552, 2010.

\bibitem{Liu2014}
Y.~Liu, E.~K.~P. Chong, A.~Pezeshki, and B.~Moran.
\newblock Bounds for approximate dynamic programming based on string
  optimization and curvature.
\newblock In {\em Proceedings of the 53rd IEEE Conference on Decision and
  Control}, 2014.

\bibitem{Minko2015}
T.~Minko, R.~Wisniewski, J.~D. Bendtsen, and R.~Izadi-Zamanabadi.
\newblock Cost efficient optimization based supervisory controller for
  supermarket subsystems with heat recovery.
\newblock In {\em Proceedings of 2015 European Control Conference}, 2015.

\bibitem{Navigant2009}
{Navigant Consulting, Inc.}
\newblock Energy savings potential and {R\&D} opportunities for commercial
  refrigeration.
\newblock Technical report, {U.S. Department of Energy}, 2009.

\bibitem{Nemhauser1978}
G.~L. Nemhauser, L.~A. Wolsey, and M.~L. Fisher.
\newblock An analysis of approximations for maximising submodular set functions
  -- {I}.
\newblock {\em Mathematical Programming}, 265--294, 1978.

\bibitem{Pedersen2013}
R.~Pedersen, J.~Schwensen, S.~Sivabalan, C.~Corazzol, S.~E. Shafiei,
  K.~Vinther, and J.~Stoustrup.
\newblock Direct control implementation of a refrigeration system in smart
  grid.
\newblock In {\em Proceedings of 2013 American Control Conference}, 2013.

\bibitem{Piette2008}
M.~A. Piette, S.~Kiliccote, and G.~Ghatikar.
\newblock Design and imple- mentation of an open, interoperable automated
  demand response infrastructure.
\newblock Technical report, Lawrence Berkeley National Laboratory, 2008.

\bibitem{Rasmussen2012}
B.~P. Rasmussen.
\newblock Dynamic modeling for vapor compression systems--part i: Literature
  review.
\newblock {\em HVAC\&R Research}, 18(5):934--955, 2012.

\bibitem{Rasmussen2012b}
B.~P. Rasmussen.
\newblock Dynamic modeling for vapor compression systems--part ii: Simulation
  tutorial.
\newblock {\em HVAC\&R Research}, 18(5):956--973, 2012.

\bibitem{Rasmussen2006}
B.~P. Rasmussen and A.~G. Alleyne.
\newblock Dynamic modeling and advanced control of air conditioning and
  refrigeration systems.
\newblock Technical report, Air Conditioning and Refrigeration Center,
  University of Illinois at Urbana-Champaign, 2006.

\bibitem{Sarabia2009}
D.~Sarabia, F.~Capraro, L.~F.~S. Larsen, and C.~{de Prada}.
\newblock Hybrid {NMPC} of supermarket display cases.
\newblock {\em Control Engineering Practice}, 17:428--441, 2009.

\bibitem{Shafiei2013}
S.~E. Shafiei, R.~Izadi-Zamanabadi, H.~Rasmussen, and J.~Stoustrup.
\newblock A decentralized control method for direct smart grid control of
  refrigeration systems.
\newblock In {\em Proceedings of the 52rd IEEE Conference on Decision and
  Control}, 2013.

\bibitem{Shafiei2013b}
S.~E. Shafiei, H.~Rasmussen, and J.~Stoustrup.
\newblock Modeling supermarket refrigeration systems for demand-side
  management.
\newblock {\em Energies}, 6(2):900--920, 2013.

\bibitem{Sioshansi2009}
R.~Sioshansi and W.~Short.
\newblock Evaluating the impacts of real time pricing on the usage of wind
  power generation.
\newblock {\em IEEE Transaction on Power Systems}, 24(2):516--524, 2009.

\bibitem{Sonntag2008}
C.~Sonntag, A.~Devanathan, and S.~Engell.
\newblock Hybrid {NMPC} of a supermarket refrigeration system using sequential
  optimization.
\newblock In {\em Proceedings of 17th IFAC World Congress}, 2008.

\bibitem{DOE2012}
{U.S. Department of Energy}.
\newblock Commercial real estate energy alliance.
\newblock 2012.

\bibitem{Vinther2015}
K.~Vinther, H.~Rasmussen, J.~Stoustrup, and A.~G. Alleyne.
\newblock Learning-based precool algorithms for exploiting foodstuff as thermal
  energy reserve.
\newblock {\em IEEE Transactions on Control Systems Technology},
  23(2):557--569, 2015.

\bibitem{Yang2015}
I.~Yang.
\newblock {\em Risk Management and Combinatorial Optimization for Large-Scale
  Demand Response and Renewable Energy Integration}.
\newblock PhD thesis, University of California, Berkeley, 2015.

\bibitem{Yang2016}
I.~Yang, S.~A. Burden, R.~Rajagopal, S.~S. Sastry, and C.~J. Tomlin.
\newblock Approximation algorithms for optimization of combinatorial dynamical
  systems.
\newblock {\em IEEE Transactions on Automatic Control}, 2016.

\bibitem{Yang2015acc}
I.~Yang, D.~S. Callaway, and C.~J. Tomlin.
\newblock Indirect load control for electricity market risk management via
  risk-limiting dynamic contracts.
\newblock In {\em Proceedings of 2015 American Control Conference}, pages
  3025--3031, 2015.

\bibitem{Yang2011}
Z.~Yang, K.~B. Rasmussen, A.~T. Kieu, and R.~Izadi-Zamanabadi.
\newblock Fault detection and isolation for a supermarket refrigeration
  system-part one: {Kalman}-filter-based methods.
\newblock In {\em Proceedings of 18th IFAC World Congress}, 2011.

\end{thebibliography}
\vfill\eject

\end{document}